%% file: MontonePF.tex
\documentclass[letterpaper, 10 pt, conference]{ieeeconf}

 \usepackage{graphicx}
\usepackage{color,url}
\input{Packages.tex}
\input{Definitions.tex}

\input{Notation.tex}

\usepackage{hyperref}
\IEEEoverridecommandlockouts                              

\graphicspath{ {Figures/} }

\title{\LARGE \bf
A differential analysis of the power flow equations
}

\author{Krishnamurthy Dvijotham$^{1}$, Michael Chertkov$^{2}$ and Steven Low$^{1}$
\thanks{$^{1}$Krishnamurthy Dvijotham and Steven Low are with the Department of Computing and Mathematical Sciences, 
        California Institute of Technology, 1200 E California Blvd, Pasadena, CA, USA
        {\tt\small dvij@cs.washington.edu,slow@caltech.edu}}%
\thanks{$^{2}$Michael Chertkov is with the Center for Nonlinear Studies (CNLS) and T-Division, Los Alamos National Laboratory, Los Alamos, NM, USA,  {\tt\small chertkov@lanl.gov}}}

\begin{document}

\maketitle
\thispagestyle{empty}
\pagestyle{empty}

\begin{abstract}
The AC power flow equations are fundamental in all aspects of power systems planning and operations. They are routinely solved using Newton-Raphson like methods. However, there is little theoretical understanding of when these algorithms are guaranteed to find a solution of the power flow equations or how long they may take to converge. Further, it is known that in general these equations have multiple solutions and can exhibit chaotic behavior. In this paper, we show that the power flow equations can be solved efficiently provided that the solution lies in a certain set. We introduce a family of convex domains, characterized by Linear Matrix Inequalities, in the space of voltages such that there is at most one power flow solution in each of these domains. Further, if a solution exists in one of these domains, it can be found efficiently, and if one does not exist, a certificate of non-existence can also be obtained efficiently. The approach is based on the theory of monotone operators and related algorithms for solving variational inequalities involving monotone operators. We validate our approach on IEEE test networks and show that practical power flow solutions lie within an appropriately chosen convex domain.
\end{abstract}

\input{Intro.tex}

\input{TechIntro.tex}

\input{EnergyFunctionConvexity.tex}

\input{RelatedWork.tex}

\input{Numerical.tex}

\section{Conclusions and Future Work}

We have presented a novel approach to solving the PF equations based on monotone operator theory. Our main technical contribution is a characterization of the family of domains over which the PF operator is monotone. Within any of these domains, there can be at most one power flow solution. The approach leads to efficient algorithms for determining existence and actually finding solutions to the PF equations within the monotonicity domains. In spite of the progress reported, there still remain many unanswered questions. First, the numerical experiments show that our approach to selecting $\Pa$ based on solving \eqref{eq:ConvFeas} is potentially conservative and misses some solutions that could be captured by other choices of $\Pa$. Figuring out the relationship between the monotone-operator approach developed h`ere and the energy function approach of \cite{DjEnergyFun} is another important direction for future work. Finally, developing algorithms that scale well and allow these techniques to be applied to larger systems is another item in our path forward.
\section*{Acknowledgments}
The work at LANL was carried out under the auspices of the National Nuclear Security Administration of the U.S. Department of Energy at Los Alamos National Laboratory under Contract No. DE-AC52-06NA25396 and it was partially supported by DTRA Basic Research Project $\#10027-13399$. The authors also acknowledge partial support of the Advanced Grid Modeling Program in the US Department of Energy Office of Electricity.
\bibliography{Ref}
\bibliographystyle{IEEEtran}
\input{Appendix.tex}

\end{document}

%% file: Packages.tex
\usepackage{verbatim}
\usepackage{algorithm}
\usepackage{bm}
\usepackage{algpseudocode}
\usepackage[all]{xy}
\usepackage{url}
\usepackage{amsmath,amsthm,amssymb}
\usepackage{graphicx}
\usepackage{subcaption}

\newtheorem{theorem}{Theorem}[section]

\newtheorem{lemma}{lemma}
\theoremstyle{definition}
\newtheorem{definition}{Definition}

\theoremstyle{remark}
\newtheorem{remark}{Remark} 

%% file: Definitions.tex
\newcommand{\opt}[1]{{#1}^\ast}

\DeclareMathOperator*{\maxi}{\text{Maximize}}

\newcommand{\expb}[1]{\exp\left({#1}\right)}
\newcommand{\logb}[1]{\log\left({#1}\right)}
\newcommand{\norm}[1]{\left\lVert {#1} \right\rVert}

\newcommand{\powb}[2]{{\left({#1}\right)}^{#2}}

\newcommand{\tran}[1]{{#1}^T}

\newcommand{\inner}[2]{\left\langle {#1},{#2} \right\rangle}

\newcommand{\Scal}{\mathcal{S}}

\newcommand{\ic}{\mathbf{j}}

\newcommand{\tranb}[1]{{\left({#1}\right)}^T}
\newcommand{\br}[1]{\left({#1}\right)}

\newcommand{\inv}[1]{{\left(#1\right)}^{-1}}

\newcommand{\diag}{\mathrm{di}}
\newcommand{\diagb}[1]{\mathrm{di}\left({#1}\right)}

\newcommand{\R}{\mathbb{R}}

\newcommand{\C}{\mathcal{C}}

\newcommand{\FHinf}[1]{q_\infty}
\newcommand{\FHtwo}[1]{q_2}
\newcommand{\FHone}[1]{q_1}

%% file: Notation.tex
\newcommand{\E}{\mathcal{E}}
\newcommand{\G}{\mathrm{pv}}
\newcommand{\V}{\mathcal{V}}
\newcommand{\Vs}{\mathcal{V}}

\newcommand{\Lo}{\mathrm{pq}}

\newcommand{\Vset}{v}

\newcommand{\Int}{\mathrm{Int}}
\newcommand{\Sb}{\mathcal{S}}
\newcommand{\Mat}[3]{{\left[{#1}\right]}_{#2}^{#3}}
\newcommand{\Rep}[1]{\mathrm{Re}\br{#1}}
\newcommand{\Imp}[1]{\mathrm{Im}\br{#1}}
\newcommand{\Jac}[2]{\nabla {#1}\br{#2}}
\newcommand{\smi}{s}

\newcommand{\Vx}{V^x}
\newcommand{\Vy}{V^y}
\newcommand{\Vc}{V^c}

\newcommand{\PQ}{(P,Q)}
\newcommand{\PV}{(P,V)}
\newcommand{\herm}[1]{{#1}^\ast}

\newcommand{\Vref}{V_0}

\newcommand{\Com}{\mathbb{C}}

\newcommand{\Sym}[1]{\mathrm{Sy}\br{#1}}

\newcommand{\Ind}[1]{\mathbb{I}\br{#1}}

\newcommand{\stc}[2]{\begin{pmatrix}
	{#1}\\{#2}
\end{pmatrix}}

\newcommand{\Conv}[1]{\text{Conv}\br{#1}}

\newcommand{\Cop}{\C_{op}}

\newcommand{\Pa}{W}
\newcommand{\Pat}{\tilde{W}}
\newcommand{\Vo}{v}
\newcommand{\Vox}{v^x}
\newcommand{\Voy}{v^y}

%% file: Intro.tex
\section{Introduction}

Power systems are experiencing revolutionary changes due to various factors, including: Integration of renewable generation, distributed generation, smart metering, direct or price-based load-control capabilities. While potentially contributing to the long-term sustainability of the power grid, these developments also pose significant operational challenges by making the power system inherently stochastic and inhomogeneous. As these changes become more widespread, the system operators will no longer have the luxury of large positive and negative reserves. Moreover, operating the future power grid will require developing new computational tools that can assess the system state and its operational margins more accurately and efficiently than current approaches. Specifically, these new techniques need to go beyond linearized methods of analysis and ensure that the power system is  safe even in the presence of large disturbances and uncertainty. In this paper, we focus on the fundamental equations of the power system -- the Power Flow (PF) equations. The equations constitute a system of nonlinear equations and are known to exhibit complex and chaotic behavior \cite{araposthatis1981analysis}\cite{varaiya1992bifurcation}. In the past, when the power systems were operated well within their security margins, the PF equations were solved without difficulty using standard techniques like Newton-Raphson and its variants. However, changes in the power system mentioned above mean that this may no longer be possible. Thus, Newton-Raphson techniques which rely on good initialization may fail to converge. In such a situation, it becomes difficult to assess whether the system is actually operationally unsafe or the Newton-Raphson method failed because of numerical difficulties or bad initialization. In this paper, we propose an approach to remedy this problem. Our approach is based on the theory of monotone operators. Just as a convex optimization problem can be solved efficiently, one can find zeros of a monotone operator efficiently. (In fact, the gradient of a convex function is a monotone operator.) Thus, if we can show that the nonlinear PF equations can be described by a monotone operator, then they can be solved efficiently. It turns out that the PF operator is not globally monotone, however it is monotone over restricted domains.

Our main technical result is a description of the domain over which the power flow operator (whose zeros are solutions to the PF equations), is monotone. In fact, we introduce a family of monotonicity domains, parameterized by a matrix-valued parameter. Each monotonicity domain is characterized by a Linear Matrix Inequality (LMI) in the real and imaginary components of the voltage phasor at each bus. Within each of these monotonicity domains, there can be at most one solution of the PF equations. Further, the solution of the PF equations within each domain can be reduced to the solution of a \emph{monotone variational inequality}, for which there exist efficient algorithms. The algorithms run in polynomial time and terminate either by finding a solution to the PF equations or a certificate that no solution exists to the PF equations within that domain.

The choice of monotonicity domain is critical, since different domains cover different parts of the voltage space. Ideally, one would like to find a monotonicity domain that covers all the solutions of interest, i.e. solutions that satisfy typical operational constraints on voltages and phases. In general, this is a hard problem. We deal with this problem by proposing a sufficient condition to ensure that voltages satisfying some bounds are constrained within the monotonicity domain. Finding the largest bound can then be cast as a quasi-convex optimization problem. We use this technique to pick the monotonicity domain. Numerical tests show that this approach is able to generate a sufficiently large monotonicity domain for several test networks.

The rest of this paper is organized as follows. Section \ref{sec:Intro} covers relevant background on power systems and monotone operators. The main technical results are presented in Section \ref{sec:Montone}. In Section \ref{sec:Related}, we discuss how our approach compares to related work on conditions for existence and uniqueness of PF solutions. In Section \ref{sec:Num}, we present numerical results illustrating our approach on some IEEE benchmark networks.

%% file: TechIntro.tex

\section{Modeling Power Systems}\label{sec:Intro}

\newcommand{\neb}{\sim}
\subsection{Notation}
$\R$ is the set of real numbers, $\Com$ the set of complex numbers. $\R^n,\Com^n$ denote the corresponding Euclidean space in $n$ dimensions.

Given a set $\C\subset\R^n$, $\Int\br{\C}$ denotes the interior of the set. Given a complex number $x\in\Com$, $\Rep{x}$ denotes its real part and $\Imp{x}$ its imaginary part. $\herm{x}$ denotes its complex conjugate. If $X\in\Com^{n\times n}$ is a square matrix with complex entires, $\herm{X}$ denotes the conjugate transpose. $\norm{x}$ refers to the Euclidean norm of a vector $x\in \R^n$ or $x\in\Com^n$ and $\inner{x}{y}$ to the standard Euclidean dot product. Given an vector $x\in\R^n$, $\diagb{x}$ denotes the $n\times n$ diagonal matrix with $\br{i,i}$-th entry equal to $x_i$. $\mathcal{S}^k$ denotes the space of $k\times k$ symmetric matrices.

Given a differentiable function $f:\R^k\mapsto\R^k$, $\nabla f$ denotes the Jacobian of $f$, a $k\times k$ matrix with the $i$-th row being the gradient of the $i$-th component of $f$.

For $M\in\R^{n\times n}$, $\Sym{M}=\frac{M+\tran{M}}{2}$. $\Ind{p}$ is the indicator function:
\begin{align*}
\Ind{p}=\begin{cases}
	1 & \text{ if } p=\text{True}\\
	0 & \text{ if } p=\text{False}\\
\end{cases}	
\end{align*}

\subsection{Power System Model}

The transmission network is modeled as a graph $\br{\Vs,\E}$ where $\Vs$ is the set of nodes and $\E$ is the set of edges. In power systems parlance, the nodes are called buses and the edges are called lines (transmission lines). We shall use these terms interchangeably in this manuscript. Nodes are denoted by indices $i=0,1,\ldots,n$ and edges by ordered pairs of nodes $\br{i,j}$. We pick an arbitrary orientation for each edge, so that for an edge between $i$ and $j$, only one of $\br{i,j}$ and $\br{j,i}$ is in $\E$. If there is an edge between buses $i$ and $j$, we write $i \neb j,j\neb i$.

 Edges correspond to transmission lines. The transmission network is characterized by its complex admittance matrix $Y \in \Com^{n\times n}$. $Y$ is symmetric but not necessarily Hermitian. Let $G=\Rep{Y},B=\Imp{Y}$.

Let $V_i$ be the voltage phasor, $P_i$ and $Q_i$ denote active and reactive injection at the bus $i$ respectively. $V$ is the vector of voltage phasors at all buses. Buses are of three types:
\begin{itemize}
\item\underline{\PV~ buses} where active power injection and voltage magnitude are fixed, while voltage phase and reactive power are variables. The set of \PV~ buses is denoted by $\G$. The voltage magnitude setpoint at bus $i\in\G$ is denoted by $\Vset_i$.
\item\underline{\PQ~ buses} where active and reactive power injections are fixed, while voltage phase and magnitude are variables. The set of \PQ~ buses is denoted by $\Lo$.
\item\underline{Slack bus}, a reference bus at which the voltage magnitude and phase are fixed, and the active and reactive power injections are free variables. The slack bus is denoted by $\Sb$ and its voltage phasor by $\Vref$. We choose bus $0$ as the slack bus as a convention.
\end{itemize}
\subsection{Background}
\subsubsection{Power Flow Equations}
The PF equations model the flow of power over the network. They are a set of coupled nonlinear equations that follow from Kirchoff's laws applied to the AC power network. Circuit elements in the standard power systems models are all linear, if one ignores discrete elements like tap-changing transformers. Even though Ohm's laws and Kirchhoff's laws are linear in voltages and currents, power is a product  of a voltage and a current and hence quadratic.  

PF equations are static and as such the equations model the regime when the network is balanced, that is, the net sum of power consumptions, injections and power dissipated is zero. This relies on the assumption that at the time-scale where the PF equations are solved (every few minutes), the system is in a quasi-steady state, i.e. the dynamic disturbances have been resolved through actions of the automatic control (voltage regulators, power system stabilizers and primary and secondary frequency control systems).
The complete set of PF equations over the graph $\br{\Vs,\E}$ is formally stated as
\begin{align*}
 P_i & =\Rep{V_i{\herm{\br{YV}}_i}} \quad i\in\Lo\cup\G\\
Q_i &=\Imp{V_i{\herm{\br{YV}}_i}}\quad i\in\Lo\nonumber \\
|V_i| & = \Vset_i \quad i\in\G\\
V_{\Sb} & = \Vref
\end{align*}
It will also be convenient, for what follows, to utilize the Cartesian parametrization of voltages: \[V_i=\Vx_i+\ic\Vy_i, i\in\Lo\cup\G.\]  
Let $\Vc=\begin{pmatrix}\Vx \\ \Vy\end{pmatrix}$ denote the stacked vector of real and imaginary voltage components, so that $\Vc_i=\Vx_i,\Vc_{n+i}=\Vy_i$.
\begin{definition}
Define the power flow operator as $F:\R^{2n}\mapsto \R^{2n}$
\begin{subequations}
\begin{align}
[F\br{\Vc}]_i & =G_{ii}\br{\powb{\Vx_i}{2}+\powb{\Vy_i}{2}}\nonumber\\
& \,-\sum_{j \neb i} B_{ij}\br{\Vy_i\Vx_j-\Vx_i\Vy_j}\nonumber \\
& -\sum_{j\sim i} G_{ij}\br{\Vx_i\Vx_j+\Vy_i\Vy_j}-P_i ,i\in\G\cup\Lo\label{eq:F2a}\\
[F\br{\Vc}]_{n+i} &=B_{ii}\br{\powb{\Vx_i}{2}+\powb{\Vy_i}{2}}
\nonumber\\ & +\sum_{j\sim i}B_{ij}\br{\Vx_i\Vx_j+\Vy_i\Vy_j}\nonumber\\
&+\sum_{j\sim i}G_{ij}\br{\Vy_i\Vx_j-\Vx_i\Vy_j}-Q_i ,i\in\Lo\label{eq:F2b} \\
[F\br{\Vc}]_{n+i} &=\powb{\Vx_i}{2}+\powb{\Vy_i}{2}-\Vset_i^2,i\in\G\label{eq:F2c}
\end{align}\label{eq:F}	
\end{subequations}
Then the PF equations can be written as
\[F\br{\Vc}=0.\]
\end{definition}

\subsection{Monotone Operators}\label{sec:Mon}

We now review briefly the theory of monotone operators, as is relevant to the approach developed in this paper. For details and proofs of the results quoted in this Section, we refer the reader to the recent survey \cite{boydmonotone}.
A function $H:\R^k\mapsto\R^k$ is said to be a monotone operator over a convex domain $\C$ if
\[\inner{H\br{x}-H\br{y}}{x-y}\geq 0 \quad \forall x,y\in\C\]
A monotone operator is a generalization of a monotonically increasing function (indeed, if $k=1$, the above condition is equivalent to monotone increase: $x\geq y \implies H\br{x}\geq H\br{y}$).
$H$ is said to be strongly monotone with modulus $m$ (or simply strongly monotone) if
\[\inner{H\br{x}-H\br{y}}{x-y}\geq \frac{m}{2}\norm{x-y}^2 \quad \forall x,y\in\C\]
for some $m>0$.
A common example of a monotone operator is the gradient of a differentiable convex function. 
\begin{definition}[Monotone Variational Inequality]
Let $\C\subset \R^k$ be a convex set and $H$ be a monotone operator over $\C$. The variational inequality (VI) problem associated with $H$ and $\C$ is:
\begin{align}
\text{ Find } x\in\C \text{ such that } \inner{H\br{x}}{y-x}\geq 0\quad \forall y\in\C\label{eq:VI}
\end{align}
\end{definition}

The following result shows that monotone variational inequalities with compact domains always have a solution and can be solved efficiently. 
\begin{theorem}\label{thm:VI}
If $H$ is strongly monotone operator over a compact domain $\C$, then \eqref{eq:VI} has a unique solution $\opt{x}$. Further, an approximate solution $x_\epsilon\in\C$ satisfying
\begin{align}
 \norm{x_\epsilon-x}\leq \epsilon\label{eq:VIapprox}
\end{align}
can be found using at most $O\br{\logb{\frac{1}{\epsilon}}}$ evaluations of $H$ and projections onto $\C$.
 \end{theorem}
 \begin{remark}
In this manuscript, we are interested in finding zeros of the PF operators introduced above. 
We can use monotone operator theory for this as follows: Suppose $H$ satisfies the hypotheses of theorem \ref{thm:VI}. If there exists a point $\opt{x}\in\C$ with $H\br{\opt{x}}=0$, then this is the unique solution of the variational inequality. Conversely, if the variational inequality has a solution with $H\br{\opt{x}}\neq 0$, then have a certificate that there is no solution of $H\br{x}=0$ with $x\in\C$.
\end{remark}

The next result provides a simple characterization of monotonicity for differentiable operators:
 \begin{theorem}\label{thm:VIcond}
 Suppose that $H$ is differentiable. Then $H$ is strongly monotone with modulus $m>0$ over $\C$ if and only if
\[\Jac{H}{x}+\tran{\Jac{H}{x}}\succeq m I \quad \forall x \in \C\] 	
 \end{theorem}

%% file: EnergyFunctionConvexity.tex
\section{Monotonicity of the Power Flow Operator}\label{sec:Montone}

In this Section, we study the monotonicity of the PF operator $F$ \eqref{eq:F}. As described in Section \ref{sec:Mon}, zeros of $F$ (solutions to the PF equations) can be found efficiently if $F$ is monotone. Thus, if we can prove that the PF operator is monotone, the PF solutions can be found efficiently. Since PF equations can have multiple isolated solutions, it is not possible that the PF operator is globally monotone because this would imply a unique solution to the PF equations. Thus, we focus on characterizing domains over which the PF operator (or a scaled version of it) is monotone. Proofs of all theorems in this section are deferred to the appendix Section \ref{sec:App}. 

\subsection{Characterization of Domain of Monotonicity of the Power Flow Operator}

We now derive our main results on the monotonicity of the PF operator \eqref{eq:F}. In order to state the result succinctly, we will need to define some matrices that are functions of the network topology, locations of \PV~buses and of the network impedance matrices.

\begin{definition}
Define the following row vectors:\begin{align*}
G^i &=\begin{pmatrix}G_{i1},G_{i2},\ldots,G_{in} \end{pmatrix}	\\
B^i &=\begin{pmatrix}B_{i1},B_{i2},\ldots,B_{in} \end{pmatrix}	\\
G^i_{\Lo} &=\begin{pmatrix}G_{i1}\Ind{1\in\Lo},\ldots,G_{in}\Ind{n\in\Lo} \end{pmatrix}\	\\
B^i_{\Lo} &=\begin{pmatrix}B_{i1}\Ind{1\in\Lo},\ldots,B_{in}\Ind{n\in\Lo} \end{pmatrix}	
\end{align*}
 Let $e_i$ denote the $i$-th column of the $n\times n$ identity matrix, so $e_i$ has zeros everywhere except the $i$-th entry. 	
\end{definition}
\begin{definition}
Define the following matrices:
\begin{subequations}
\begin{align}
M_i &= \begin{cases}
\begin{pmatrix}
 \diagb{G^i}+e_iG^i & \diagb{B^i}-e_iB^i	\\
 -\diagb{B^i_{\Lo}}-e_iB^i_{\Lo} & \diagb{G^i_{\Lo}}-e_iG^i
 \end{pmatrix} & i\in\Lo\\
\begin{pmatrix}
 \diagb{G^i}+e_iG^i & \diagb{B^i}-e_iB^i	\\
 -\diagb{B^i_{\Lo}}+2e_i\tran{e_i} & \diagb{G^i_{\Lo}}
 \end{pmatrix} &  i\in\G 	
 \end{cases} \\
N_i &= \begin{cases}
\begin{pmatrix}
 -\diagb{B^i}+e_iB^i & \diagb{G^i}+e_iG^i	\\
 -\diagb{G^i_{\Lo}}+e_iG^i_{\Lo} & -\diagb{B^i_{\Lo}}-e_iB^i
 \end{pmatrix} &  i\in\Lo\\
\begin{pmatrix}
 -\diagb{B^i}+e_iB^i & \diagb{G^i}+e_iG^i	\\
 -\diagb{G^i_{\Lo}} & -\diagb{B^i_{\Lo}}+2e_i\tran{e_i}
 \end{pmatrix} &  i\in\G 	
 \end{cases}  \\
 M_0 &= \begin{pmatrix}
 \diagb{G^0} & \diagb{B^0}	\\
 -\diagb{B^0_{\Lo}} & \diagb{G^0_{\Lo}}
 \end{pmatrix} \\
 N_0 &= \begin{pmatrix}
-\diagb{B^0}& \diagb{G^0}\\
 -\diagb{G^0_{\Lo}} & -\diagb{B^0_{\Lo}}
 \end{pmatrix}
\end{align}		
\end{subequations}
\end{definition}

We now state our main technical result, which shows that the PF equations can be solved by solving a monotone variational inequality (for which there exist efficient algorithms).

\begin{theorem}\label{thm:Monotone}
Let $\Pa \in \R^{2n \times 2n}$ be an invertible matrix and $m>0$. There is at most one solution of the PF equations $F\br{\Vc}=0$ over the domain:
\begin{align}
\C\br{\Pa}=\left\{\begin{pmatrix}\Vx\\\Vy\end{pmatrix}: \sum_{i=0}^{n} \Sym{\Pa\br{M_i\Vx_i+N_i\Vy_i}}\succeq m I\right\}
\end{align}
Define $F_{\Pa}\br{\Vc}=\Pa F\br{\Vc}$. $F_{\Pa}$ is strongly monotone with modulus $m$ over the set $\C\br{\Pa}$. Let $\opt{\Vc}$ be the unique solution to the monotone variational inequality:
\begin{align}
&\Vc\in\C\br{\Pa}\\
&\inner{F_{\Pa}\br{\Vc}}{\alpha-\Vc}\geq 0 \quad\forall \alpha:\alpha\in\C\br{\Pa}
\end{align}
If $F\br{\opt{\Vc}}=0$, $\opt{\Vc}$ is the unique solution to the PF equations in $\C\br{\Pa}$. Otherwise, there are no solutions to $F\br{\Vc}=0,\Vc\in\C\br{\Pa}$.
\end{theorem}
\begin{remark}
We implicitly assume existence of a solution to the variational inequality above. This is guaranteed (by theorem \ref{thm:VI}) if $\C\br{\Pa}$ is compact. Magnitudes of voltage phasors in practical power systems are always bounded, and one can simply pick a large positive number $b>0$ such that all power flow solutions satisfy  $\norm{\Vc}\leq b$ and define $\C\br{\Pa}$ with this additional constraint, ensuring it is a compact set.
\end{remark}

\subsubsection{Illustration on 2-bus network}

We consider a 2-bus network with admittance matrix
\[Y=\begin{pmatrix}
.05-\ic 1.11 & -.05+\ic 1.11 \\
-.05+\ic 1.11  & .05-\ic 1.11 	
\end{pmatrix}\]

We fix $=I,V_0=1+\ic 0$. Let $V_1-V_0=\Vx+\ic \Vy$. By varying $\Pa$, we can find two disjoint monotonicity domains for this system:
\begin{align*}
\Pa_1=\begin{pmatrix}
-6.75 & .31 \\
 0    & .1 	
\end{pmatrix}, \Pa_2=\begin{pmatrix}
-6.75 & .31 \\
 0    & -.1 	
\end{pmatrix}		
\end{align*}

It can be verified that these two monotonicity domains cover the space of $\br{\Vx,\Vy}$ and there is at most one solution in each domain. In fact, in this system, as long as there is a solution, there exist two (one in each domain). 
As the loading increases, the two solutions move closer to each other and eventually disappear (at the point of the saddle-node bifurcation). Note that the boundary between the two solutions is observed exactly at $V_1-V_0=-\frac{1}{2}$. Note that letting $V_1=V\expb{\ic\theta}$, or equivalently
\[V\cos\br{\theta}\geq \frac{1}{2}\]
which is a well-known classical criterion for voltage stability in the two bus example.

\subsection{Selection of the Monotonicity Domain}\label{sec:MonCC}

Theorem \ref{thm:Monotone} shows that PF solutions in $\C\br{\Pa}$ can be found efficiently. However, the matrix inequality characterizing $\C\br{\Pa}$ is not intuitive and does not have a simple interpretation. Further,  the choice of $\Pa$ would affect the domain of monotonicity. Thus, it is important to pick $\Pa$ so that $\C\br{\Pa}$ contains the ``operationally relevant'' PF solutions. In this Section, we describe a technique  to achieve this goal.

We consider the following class of operational constraints:
\begin{align*}
	|V_i-\Vo_i|\leq \delta_{i}
\end{align*}
where $\Vo$ is a ``nominal voltage profile. Intuitively, one can think of this as the ``center'' of the monotonicity domain.  Stated in terms of $\Vc=\begin{pmatrix}
	\Vx\\ \Vy
\end{pmatrix}$, the constraint becomes
\begin{align}
	&\sqrt{|\Vx_i-{\Vox_i}|^2+|\Vy_i-{\Voy_i}|^2}\leq \delta_{i} \quad i=1,\ldots,n\label{eq:Flim}\nonumber\\
	&\Cop\br{\delta}=\nonumber\\
	&\left\{\begin{pmatrix}
	\Vx\\ \Vy
\end{pmatrix}: \sqrt{|\Vx_i-{\Vox_i}|^2+|\Vy_i-{\Voy_i}|^2}\leq \delta_{i}\quad i=1,\ldots,n\right\}
\end{align}

We now state a theorem that gives a sufficient condition to guarantee that $\Cop\br{\delta}\subset\C\br{\Pa}$.

\begin{theorem}\label{thm:MonotoneC}
Let
\begin{align}
K = \begin{pmatrix}
	1-\sqrt{2} & -1 &-1 & 1-\sqrt{2} \\
	1          & \sqrt{2}-1 & 1-\sqrt{2} & -1
\end{pmatrix}	
\end{align}
	Suppose $\exists \Pa,X_1,\ldots,X_n\in\mathcal{S}^{2n}$ satisfying
	\begin{subequations}
	\begin{align}
		\sum_{i=0}^n\Sym{\Pa\br{ M_i\Vox_i+N_i\Voy_i}}\succeq mI+\sum_{i=1}^n X_i\label{eq:CCa} \\
		X_i \succeq \delta_k\br{-K_{1l} M_i-K_{2l} N_i}, l\in\{1,\ldots,4\}\label{eq:CCb}\\
		X_i \succeq \delta_k\br{K_{1l} M_i-K_{2l} N_i}, l\in\{1,\ldots,4\}\label{eq:CCc}
			\end{align}\label{eq:CC}		
	\end{subequations}
	Then $\Cop\br{\delta}\subset \C\br{\Pa}$. Further, this is a convex feasibility problem in $\Pa$, $X_1,\ldots,X_n$. 	
\end{theorem}


Theorem \ref{thm:MonotoneC} gives us a principled way to choose $\Pa$ so that $\Cop\subset\C\br{\Pa}$. However, the system of constraints \eqref{eq:CC} may be infeasible. In this situation, we can try to see if a scaled version of the constraint set $\C\br{\delta}$ lies within $\C\br{\Pa}$. Specifically, we look for the largest possible $\rho>0$ such that $\C\br{\rho\delta}\subseteq\C\br{\Pa}$.

We then formulate the following problem:
\begin{subequations}
\begin{align}
	& \maxi_{\rho\geq 0,\Pa\in\R^{2n\times 2n},X_i \in\Scal^{2n}} \quad  \rho \\
	&\text{Subject to } \nonumber\\
&\sum_{i=0}^n\Sym{\Pa\br{ M_i\Vox_i+N_i\Voy_i}}\succeq mI+\sum_{i=1}^n X_i\label{eq:CCa} \\
&		X_i \succeq \rho\delta_k\br{-K_{1l} M_i-K_{2l} N_i}, l\in\{1,\ldots,4\}\label{eq:CCb}\\
&		X_i \succeq \rho\delta_k\br{K_{1l} M_i-K_{2l} N_i}, l\in\{1,\ldots,4\}\label{eq:CCc}	\end{align}	\label{eq:ConvFeas}	
\end{subequations}
	
\begin{lemma}\label{lem:ConvFeas}
\eqref{eq:ConvFeas} is a quasi-convex optimization problem and can be solved efficiently. Further, if the matrix
\[\sum_{i=0}^n M_i\Vox_i+N_i\Voy_i\]
is full-rank, the problem is always feasible and the optimal solution satisfies $\Cop\br{\rho\delta}\subseteq \C\br{\Pa}$.
\end{lemma}	
\begin{remark}
The condition that $\sum_{i=0}^n M_i\Vox_i+N_i\Voy_i$ is non-singular means that the Jacobian of the power flow equations at the nominal voltage profile $\Vo$	is non-singular. This is generically true, except for very special choices of the $\Vo$ and the network admittance matrix. Further, power system voltage stability criteria are often formulated in terms of the ``distance to singularity'' of the Jacobian matrix, so enforcing a non-singular Jacobian is a reasonable restriction. 
\end{remark}

%% file: RelatedWork.tex
\section{Related Work}\label{sec:Related}

Several papers have studied conditions for existence of solutions to the PF Equations \cite{bolognani2014existence}\cite{lesieutre1999existence}\cite{wu1982steady}. In \cite{molzahn2012sufficient}, the authors propose a sufficient condition for the insolvability of the PF equations based on a convex relaxation.
Our approach differs from these in the following important ways:

\begin{itemize}
\item To solve the PF equations in $\C\br{\Pa}$, we provide necessary and sufficient conditions, that is, our approach finds a solution in $\C\br{\Pa}$ if there exists one, and a certificate of non-existence if there is no solution.

\item Our approach is algorithmic, that is, we provide an algorithm (based on solving a monotone variational inequality) that is guaranteed to find the solution efficiently, i.e. in polynomial time.

\item If there are additional operational constraints $\br{\Vx,\Vy}\in S$, where $S$ is a convex set and $S\subset \C\br{\Pa}$ , e.g. corresponding to line protection \cite{singh2001direct} and/or line flow limits, we can additionally answer the question of whether there exists a PF solution with $\br{\Vx,\Vy}\in S$. This is an important contribution, since most of the time system operators are interested in finding PF solutions that additionally satisfy operational constraints.
    
\item We can find multiple PF solutions by choosing different nominal voltage profiles $\Vo$. This is relevant for a number of important power system applications, such as assessing distance to voltage collapse and transient stability (computation of the so-called controlling unstable equilibrium point).
\end{itemize}

In \cite{bolognani2014existence}, the authors also propose an algorithm based on a contraction mapping. However, the algorithm only works in a small ball around the origin in the $\br{P,Q}$ space. This was extended to other kind of sets in \cite{KostyaExist}. On the other hand, our results are stated in terms of a  convex constraint in $\br{\Vx,\Vy}$ space. Understanding the set of $\br{P,Q}$ for which the solution $\br{\Vx,\Vy}\in\C\br{\Pa}$, and conversely the set of $\br{\Vx,\Vy}$ for which $\br{P,Q}$ lies in a certain set, is still an open problem, even for the special case where all buses are \PV~buses. This setting was studied and the results were extended and connected with results on synchronization in coupled oscillators in a series of recent papers \cite{dorfler2012synchronization}\cite{dorfler2013synchronization}\cite{dorfler2014synchronization}. These authors provided distinct sufficient and necessary conditions on the injections for the existence of power flow solutions with phase differences satisfying certain bounds. 

In recent work \cite{DjEnergyFun}, we have shown that for the special case of lossless networks (and networks with constant ratio of inductance to resistance) the PF equations can be solved analyzing a convex optimization problem. Specifically, we minimize the so-called energy function over a restricted domain (the convexity domain of the energy function). The optimization problem in \cite{DjEnergyFun} was formulated in polar coordinates, i.e. in terms of the voltage magnitude and phase at each bus. In that setting, we get a single domain characterized by a nonlinear but convex matrix inequality, and any solution within this domain can be found efficiently. In fact, if one writes the PF equations in polar coordinates, one can show that the monotonicity domain coincides precisely with the convexity domain of the energy function. However, this polar-coordinate monotonicity domain of \cite{DjEnergyFun} is not equivalent to monotonicity domain in the Cartesian coordinates discussed in this manuscript. In this context, work reported in the present manuscript was inspired by our earlier attempt to extend the polar-coordinate based approach to lossy networks. Even though we were not able to find a simple convex characterization of the monotonicity domain in the polar coordinates, it lead us to discover that in the Cartesian coordinates the monotonicity domain can be described using LMIs. Numerical tests show that the domain of convexity of the energy function is a superset of the monotonicity domain computed by \eqref{eq:ConvFeas} for lossless networks. Furthermore, in that setting, we were able to prove a stronger result for tree networks showing that there exists a PF solution if and only if there exists one in the convexity domain. The condition characterizing the convexity domain also had a simple physical interpretation - it simply requires voltage magnitudes and phases at the neighboring buses to be ``close'' to each other.

However, the approach reported in the present manuscript is more general. It allows one to tune the monotonicity domain to compute multiple solutions. Furthermore, the monotonicity domain is expressed as an LMI here and can be handled using off-the-shelf software, while the nonlinear matrix inequality from \cite{DjEnergyFun} requires specially designed solvers. Understanding the relationship between the approach of this manuscript and the one from \cite{DjEnergyFun}, in particular, the relationship between cartesian and polar parameterizations, is an important direction for future work.


%% file: Numerical.tex
\section{Numerical Illustrations}\label{sec:Num}

This Section presents numerical experiments illustrating theoretical results presented above.  We start by illustrating monotonicity domain on a 3 bus system example, and then discuss  numerical experiments performed on the case9 and case14 systems available with the MATPOWER \cite{zimmerman2011matpower} software. For solving the convex optimization problem \eqref{eq:ConvFeas}, we use the parser-solver CVX \cite{cvx}\cite{gb08}. For solving the variational inequality, we implement our own solver based on the extra-gradient method described in \cite{boydmonotone}. The implementations used here are not optimized and currently do not scale to larger systems. Algorithmic developments that would enable applications of these ideas to larger network will be an important direction for future work.

\subsection{Illustration of Monotonicity domain for 3-bus network}

We consider a 3 bus network with bus $0$ being the slack bus, bus $1$ a \PV~bus and bus $2$ a \PQ~bus. The voltage phasor at the slack bus is taken to be $1+\ic 0$ and at the \PV~bus to be $\expb{\ic\theta}$. The phasor at the \PQ~bus is parameterized as $1+\ic 0+\Vx+\ic \Vy$ where $\Vx,\Vy$ represent the \emph{deviations} from the reference voltage. The variables to be solved for in the PF equations are $\Vx,\Vy,\theta$. We pick a nominal voltage profile by setting all voltages equal to $1+\ic 0$. We use \eqref{eq:ConvFeas} to find a monotonicity domain $\Pa$. We then plot the monotonicity domain in $\br{\Vx,\Vy}$ space for multiple values of $\theta$ in Fig.~(\ref{fig:ThreeBus}).

Fig.~(\ref{fig:ThreeBus}) shows that for small values of $\theta$, the monotonicity domain covers a fairly large space in $\Vx,\Vy$. As $\theta$ increases, the domain shrinks and ultimately, before $|\theta|$ hits $\frac{\pi}{2}$, it becomes empty. This is consistent with the idea that all voltages ``close'' to the nominal voltage profile are contained within the monotonicity domain. Fig.~(\ref{fig:ThreeBus}) also shows that the domain actually includes a fairly large region around the nominal voltage profile.

\begin{figure}[htb]
\centering
        \begin{subfigure}[b]{0.35\columnwidth}
                \includegraphics[width=.8\textwidth]{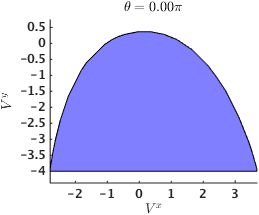}
                \caption{$\theta=0$}
                \label{fig:ThreeBusA}
        \end{subfigure}%
                \begin{subfigure}[b]{0.35\columnwidth}
                \includegraphics[width=.8\textwidth]{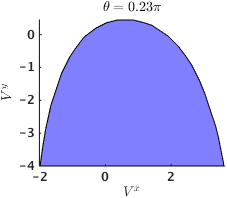}
                \caption{$\theta=.23\pi$}
                \label{fig:ThreeBusB}
        \end{subfigure}%

        \begin{subfigure}[b]{0.35\columnwidth}
                \includegraphics[width=.8\textwidth]{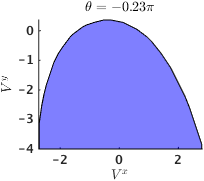}
                \caption{$\theta=-.23\pi$}
                \label{fig:ThreeBusC}
        \end{subfigure}%
        \begin{subfigure}[b]{0.35\columnwidth}
                \includegraphics[width=.8\textwidth]{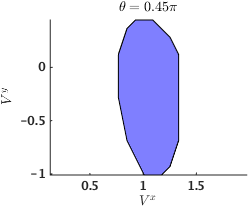}
                \caption{$\theta=.45\pi$}
                \label{fig:ThreeBusD}
        \end{subfigure}%

        \begin{subfigure}[b]{0.35\columnwidth}
                \includegraphics[width=.8\textwidth]{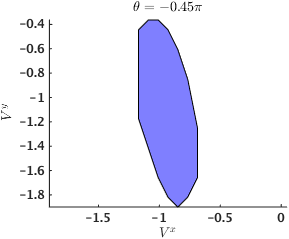}
                \caption{$\theta=-.45\pi$}
                \label{fig:ThreeBusE}
        \end{subfigure}%
\caption{Monotonicity Domain for a $3$ bus system}\label{fig:ThreeBus}
\end{figure}

\subsection{Computation of Voltage Phasor Bounds}

We use the methodology of Section \ref{sec:MonCC} to find the maximum possible deviation from a nominal voltage profile contained within a monotonicity domain. We choose $\Vo_i=V_0$, so that every bus has a voltage phasor equal to the reference at the slack bus. For a system with only \PQ~buses, this is in fact a solution to the PF equations for the case of no, $0$, injections. In practical power systems, large deviations from the slack bus voltage are relatively rare.
 Therefore, it is reasonable to seek a monotonicity domain that includes all solutions in some neighborhood of this special (nominal) one.

We choose a uniform bound $\delta_{i}=1$ and solve \eqref{eq:ConvFeas} to compute $\Pa$. We report the bounds for our test networks in the Table \ref{tab:Bound}. These bounds easily suffice to cover typical voltages observed in practical power systems operations. However, these domains may be insufficient to describe solutions in the case of extreme injections used in planning studies and contingency analysis. Our numerical experience suggests that these bounds are are conservative. More specifically, there are many $\br{\Vx,\Vy}$ that do not satisfy the bound given by $\rho$, but still lie within $\C\br{\Pa}$. The results of the next section establish this numerically.

\begin{table}[htb]
\vspace{.8em}
\begin{center}
\begin{tabular}{| c | c |}
\hline
System  & Limit on $|V_i-V_0|$ (p.u system) \\
\hline
Case 9  & .2 \\
\hline
Case 14 & .25 	\\
\hline
\end{tabular}
\caption{Bounds on $\Vc$ guaranteeing $\Vc\in\C\br{\Pa}$}\label{tab:Bound}
\end{center}
\end{table}

\subsection{Scaling Loads}
In this Section, we use the monotonicity domain obtained by solving \eqref{eq:ConvFeas}, but instead of relying on the computed bound $\rho$, we test numerically whether existence of a solution of the PF equations implies existence of a solution in the monotonicity domain. Of course, determining existence of a solution to the PF equations is a hard problem, so we instead rely on sufficient conditions for insolvability developed in \cite{molzahn2012sufficient}. We use the implementation available with the MATPOWER package \cite{molzahn2013implementation}.

Our experimental approach is as follows: We take the injections given as part of the Matpower test case, and scale each injection by a complex scalar: $\tilde{S}_i=\alpha S_i$. By scaling the magnitude of injections, we reach a point where there are no solutions to the power flow equations (maximum loadability). By scaling by a complex factor, we effectively achieve different scaling for the real and reactive parts of the injections. We pick the magnitude of $\alpha$ uniformly distributed between $1$ and a maximum value and the phase uniformly distributed between $\pi/6$ and $\pi/3$. (We observe that once we scale beyond a certain { threshold} there are very few cases where the PF equations have a solution.) These choices create a space of injections that include at least part of the solvability boundary (injections beyond which the PF equations have no solution). For a uniform grid of the $\alpha$ space described above with $100$ points, we compute the number of points at which there are no PF solutions (that is in this case we have a certificate of infeasibility based on \cite{molzahn2012sufficient}) and the number of points at which there are no solutions in $\C\br{\Pa}$. The results are shown in Table (\ref{tab:Scale}).
\begin{table}[htb]
\begin{center}
\begin{tabular}{| c | c | c | c |}
\hline
System  & $\C\br{\Pa}$ & $-\C\br{\Pa}$ & NoSol \\
\hline
Case 9 & 38 & 6 & 56 	\\
\hline
Case 14  & 40 & 49 & 11 \\
\hline
\end{tabular}
\caption{Existence of Solutions in the Monotonicity Domain: The first column denotes the number of instances (out of 100) such that a PF solution within $\C\br{\Pa}$ was found. The second denotes the number of instances for which a PF solution within $\C\br{\Pa}$ was not found, and there is no certificate of insolvability of the power flow equations based on the sufficient condition from \cite{molzahn2012sufficient}. The last column denotes the number of instance where we have a certificate of infeasibility.}\label{tab:Scale}
\end{center}
\end{table}

The results show that for the 9 bus network, most cases where a solution exists lie within the monotonicity domain. However, for the 14 bus system, there exists a significant number of cases for which the solution does not lie in the monotonicity domain, if it does exist. It turns out that there are larger monotonicity domains that contain the original one (we checked this post-hoc by solving a feasibility problem to find a $\Pa$ such that $\C\br{\Pa}$ includes all the solutions found using Newton-Raphson). Thus, it seems like the choice of $\Pa$ based on \eqref{eq:ConvFeas} is not optimal and larger monotonicity domains could be found. How one might do this in a principled way constitutes a comprehensive direction for future work.

\subsection{Comparison to Energy Function Approach}

In this Section, we consider a lossless modification of the 9 bus and 14 bus networks and perform the same experiment as described in the previous Section. This allows us to compare the approach developed in this manuscript with the energy function based approach described in \cite{DjEnergyFun}. 

The results (table \ref{tab:EF}) show that for these networks, existence of a solution seems to imply existence within $\C_{EF}$ but not within $\C_{\Pa}$. This suggests that $\C_{EF}$ is a superset of $\C_{\Pa}$ and that potentially choosing $\Pa$ based on \eqref{eq:ConvFeas} does not cover all the solutions. Closing this gap and designing approaches to ensure that $\C_{\Pa}$ would contain all the solutions in $\C_{EF}$ is a direction for future work. We anticipate that development of this approach will likely suggest improvements applicable to lossy networks as well.
\begin{table}[htb]
\begin{center}
\begin{tabular}{| c | c | c | c  | c | c |}
\hline
System  & $\C\br{\Pa}$ & $\C_{EF}$ & $-\C\br{\Pa}$ & $-\C_{EF}$ & NoSol \\
\hline
Case 9 LL & 50 & 56 & 6 & 0 & 44 	\\
\hline
Case 14 LL & 62 & 92 & 30 & 0 & 8 	\\
\hline
\end{tabular}\caption{Monotonicity vs Energy Function: The first, third and final columns have the same meaning as table (\ref{tab:Scale}). The second column is the number of instances where a PF solution within $\C_{EF}$ was found, and the final column the number of instances where a PF solution within $\C_{EF}$ was not found and there is no certificate of infeasibility.}\label{tab:EF}
\end{center}
\end{table}

%% file: Appendix.tex
\section*{Appendix}\label{sec:App}
\subsection{Proof of theorem \ref{thm:Monotone}}\label{sec:AppendPfa}
\begin{proof}
The result is a direct consequence of the strong monotonicity of $F_{\Pa}$ over $\C\br{\Pa}$. By lemma \ref{lem:Jac}, the Jacobian of $F_{\Pa}$ is given by
\[\Jac{F_{\Pa}}{\beta}=\sum_{i=0}^n \Pa\br{M_i\Vx_i+N_i\Vy_i}\]
where $\beta=\stc{\Vx}{\Vy}$. 	
$F_{\Pa}$ is strongly montone with modulus $m$ if (theorem \ref{thm:VIcond})
\[\Sym{\Jac{F_{\Pa}}{\beta}}=\sum_{i=0}^n \Sym{\Pa\br{M_i\Vx_i+N_i\Vy_i}}\succeq mI\]	
Thus, $F_{\Pa}$ is strongly monotone over the domain $\beta\in\C\br{\Pa}$. The remaining results follow from theorem \ref{thm:VI}
\end{proof}
\subsection{Proof of theorem \ref{thm:MonotoneC}}\label{sec:AppendPfb}
\begin{proof}Let $Q_0=\sum_{i=0}^{n} \Sym{\Pa\br{M_i\Vox_i+N_i\Voy_i}}$. Then,
\begin{align}
&\sum_{i=0}^{n} \Sym{\Pa\br{M_i\Vx_i+N_i\Vy_i}}=\nonumber\\
&Q_0+\sum_{i=0}^{n} \Sym{\Pa\br{M_i\br{\Vx_i-\Vox_i}+N_i\br{\Vy_i-\Voy_i}}}\label{eq:Pfa}
\end{align}
\begin{align*}
&\Vc\in \Cop\br{\delta}\equiv \sqrt{|\Vx_i-\Vox_i|^2+|\Vy_i-\Voy_i|^2}\leq \delta_{i}\\
&\implies \stc{\Vx_i-\Vox_i}{\Vy_i-\Voy_i}\in \delta_i\Conv{\left\{\stc{K_{1l}}{K_{2l}},\stc{-K_{1l}}{K_{2l}}\right\}_{1\leq l \leq 4}}	
\end{align*}
This can be seen simply by plotting the points $\delta_i\br{K_{1l}+\ic K_{l2}}$ on the complex plane and observing that their convex hull contains the disc of radius $\delta_{i}$ centered at the origin. Thus, for any $\Vc\in \Cop\br{\delta}$,
\begin{align*}
&-\br{M_i\br{\Vx_{i}-\Vox_{i}}+N_i\br{\Vy_{i}-\Voy_{i}}}\in\\
&\delta_{i}\Conv{\left\{-K_{1l}M_i-K_{2l}N_i,K_{l1}M_i-K_{2l}N_i\right\}_{l=1,\ldots,4}}
\end{align*}
From \eqref{eq:Pfa}, we now have
\begin{align*}
& \sum_{i=0}^{n} \Sym{\Pa\br{M_i\Vx_i+N_i\Vy_i}}\succeq Q_0-\sum_{k\in\E} X_k \forall \,\Vc\in \Cop
\end{align*}
By \eqref{eq:CCa}, we have $\Vc\in \Cop\br{\delta} \implies \Vc\in\C\br{\Pa}$, so that $\Cop\br{\delta}\subset\C\br{\Pa}$. 
\end{proof}

\begin{lemma}\label{lem:Jac}
The Jacobian of $F\br{\Vx,\Vy}$ is given by:
\begin{align}
\Jac{F}{\Vx,\Vy}=\sum_{i=0}^n M_i \Vx_i+N_i\Vy_i
\end{align}	

\end{lemma}
\begin{proof}
The entries of the Jacobian are given by:
\begin{align*}
\Jac{F}{\Vx,\Vy}=\begin{pmatrix}
S & T \\
O & L 		
\end{pmatrix}
\end{align*}	
with each block being $n \times n$. 
\begin{align*}
S_{ii} &= 2G_{ii}\Vx_i+\sum_{j\sim i} B_{ij}\Vy_j	-\sum_{j\sim i} G_{ij}\Vx_j\\
S_{ij} &= -B_{ij}\Vy_i-G_{ij}\Vx_i\\
T_{ii} &= 2G_{ii}\Vy_i-\sum_{j\sim i} B_{ij}\Vx_j	-\sum_{j\sim i} G_{ij}\Vy_j\\
T_{ij} &= B_{ij}\Vx_i-G_{ij}\Vy_i 
\end{align*}
The entries of $O,L$ depend on whether $i$ is a \PQ~bus or \PV~bus. For $i\in\Lo$, we have
\begin{align*}
O_{ii} &= 2B_{ii}\Vx_i+\sum_{j\sim i} B_{ij}\Vx_j	-\sum_{j\sim i} G_{ij}\Vy_j\\
O_{ij} &= B_{ij}\Vx_i+G_{ij}\Vy_i\\
L_{ii} &= 2B_{ii}\Vy_i+\sum_{j\sim i} B_{ij}\Vy_j	-\sum_{j\sim i} G_{ij}\Vx_j\\
L_{ij} &= B_{ij}\Vy_i+G_{ij}\Vx_i
\end{align*}
For $i\in\G$, we have
\begin{align*}
O_{ii} = 2\Vx_i,O_{ij} = 0,L_{ii} = 2\Vy_i,L_{ij} = 0
\end{align*}
Using these expressions, it is not hard to see that 
\begin{align*}
\Jac{F}{\Vx,\Vy}=\sum_{i=0}^n M_i\Vx_i+N_i\Vy_i
\end{align*}
\end{proof}
\subsection{Proof of lemma \ref{lem:ConvFeas}}\label{App:ConvFeas}
\begin{proof}
Suppose $\exists \Pat,X_i \in\Scal^{2n}$ such that the constraints are feasible for some $\rho$. Then, clearly, with the same choice, they are also feasible for any smaller $\rho$. Then, the sub-level sets of the problem are convex, and hence the problem is a quasi-convex optimization problem. Suppose $\sum_{i=0}^n M_i\Vox_i+N_i\Voy_i$ is full-rank. When $\rho=0$, for any fixed $\Pat$, we can choose $X_k=0$ and $T_k,S_k$ to satisfy the final equality constraints. Then, one needs to find $\Pa$ such that
\begin{align*}
&\Sym{\Pa\br{\sum_{i=0}^n M_i\Vox_i+N_i\Voy_i}}\succeq mI	
\end{align*}
This can always be done if $\sum_{i=0}^n M_i\Vox_i+N_i\Voy_i$ is full rank, since one can simply choose
\[\Pa=m\inv{\sum_{i=0}^n M_i\Vox_i+N_i\Voy_i}\]
Hence, this optimization problem is always feasible when $\sum_{i=0}^n M_i\Vox_i+N_i\Voy_i$ is full rank.
\end{proof}